\documentclass[conference]{IEEEtran}
\IEEEoverridecommandlockouts

\usepackage{cite}
\usepackage{amsmath,amssymb,amsfonts}
\usepackage{amsthm}
	\newtheorem{definition}{Definition}

	\newtheorem{theorem}{Theorem}
\usepackage{algorithmic}
\usepackage{graphicx}
\usepackage{textcomp}
\usepackage{tikz}
	\usetikzlibrary{shapes.geometric}
	\usetikzlibrary{positioning}
	\usetikzlibrary{arrows, shapes, shapes.misc, chains, scopes, matrix, calc, fit, graphs, shadings, decorations.pathreplacing}
\usepackage{enumitem}
\usepackage{balance}
	
\usepackage{xcolor}

\usepackage{textcomp}
\usepackage{stfloats}
\usepackage{graphicx}
\usepackage{cite}
\usepackage{mathtools, xfrac, bm, amsfonts, amssymb, amsmath, dsfont, stmaryrd, mathrsfs}
\let\ltxcup\cup
\let\ltxcap\cap
\let\ltxemptyset\emptyset
\let\ltxlangle\langle
\let\ltxrangle\rangle
\usepackage{mathabx}
\let\cup\ltxcup
\let\cap\ltxcap
\let\emptyset\ltxemptyset
\let\langle\ltxlangle
\let\rangle\ltxrangle
	
\usepackage[breakable]{tcolorbox}
\usepackage{amsthm}
	\newtheorem{dfn}{Definition}
	\newtheorem{lem}{Lemma}
	
\usepackage{xparse}
\usepackage{xstring}

\newcommand{\C}[1][-]{	\IfStrEqCase{#1}{%
							{-}{\mathbb{C}}%
							{-c}{\mathbb{C}_{\mu}}%
}[\mbox{\color{red}Use of illegal parameter in command \textbackslash C!}]}
\newcommand{\R}[1][-]{	\IfStrEqCase{#1}{%
							{-}{\mathbb{R}}%
							{-c}{\mathbb{R}_{\mu}}%
}[\mbox{\color{red}Use of illegal parameter in command \textbackslash R!}]}
\newcommand{\W}[1][-]{	\IfStrEqCase{#1}{%
							{-}{\mathbb{W}}%
							{-c}{\mathbb{W}_{\mu}}%
}[\mbox{\color{red}Use of illegal parameter in command \textbackslash W!}]}
\renewcommand{\AA}[1][-]{	\IfStrEqCase{#1}{%
							{-}{\mathcal{A}}%
							{-P}{\mathrm{P}\mathcal{A}}%
}[\mbox{\color{red}Use of illegal parameter in command \textbackslash AA!}]}
\newcommand{\BB}[1][-]{		\IfStrEqCase{#1}{%
							{-}{\mathcal{B}}%
							{-P}{\mathrm{P}\mathcal{B}}%
}[\mbox{\color{red}Use of illegal parameter in command \textbackslash BB!}]}
\newcommand{\LL}[1][-]{		\IfStrEqCase{#1}{%
							{-}{\mathcal{L}}%
							{-P}{\mathrm{P}\mathcal{L}}%
}[\mbox{\color{red}Use of illegal parameter in command \textbackslash LL!}]}
\newcommand{\QQ}[1][-]{		\IfStrEqCase{#1}{%
							{-}{\mathcal{Q}}%
							{-P}{\mathrm{P}\mathcal{Q}}%
}[\mbox{\color{red}Use of illegal parameter in command \textbackslash OO!}]}
\newcommand{\TT}[1][-]{		\IfStrEqCase{#1}{%
							{-}{\mathcal{T}}%
							{-P}{\mathrm{P}\mathcal{T}}%
}[\mbox{\color{red}Use of illegal parameter in command \textbackslash TT!}]}
\newcommand{\aVar}[1][-]{	\IfStrEqCase{#1}{%
							{-}{a}%
							{-t}{\hspace{0.75pt}\widetilde{a}}
							{-m}{\mathrm{a}}%
							{-b}{\bm{a}}%
							{-mb}{\bm{\mathrm{a}}}%
							{+}{A}%
							{+m}{\mathrm{A}}%
							{+b}{\bm{A}}%
							{+mb}{\bm{\mathrm{A}}}%
}[\mbox{\color{red}Use of illegal parameter in command \textbackslash bVar!}]}
\newcommand{\cVar}[1][-]{	\IfStrEqCase{#1}{%
							{-}{c}%
							{-m}{\mathrm{c}}%
							{-b}{\bm{c}}%
							{-mb}{\bm{\mathrm{c}}}%
							{+}{C}%
							{+m}{\mathrm{C}}%
							{+b}{\bm{C}}%
							{+mb}{\bm{\mathrm{C}}}%
}[\mbox{\color{red}Use of illegal parameter in command \textbackslash cVar!}]}
\newcommand{\dVar}[1][-]{	\IfStrEqCase{#1}{%
							{-}{d}%
							{-m}{\mathrm{d}}%
							{-b}{\bm{d}}%
							{-mb}{\bm{\mathrm{d}}}%
							{+}{D}%
							{+m}{\mathrm{D}}%
							{+b}{\bm{D}}%
							{+mb}{\bm{\mathrm{D}}}%
}[\mbox{\color{red}Use of illegal parameter in command \textbackslash dVar!}]}
\newcommand{\eVar}[1][-]{	\IfStrEqCase{#1}{%
							{-}{e}%
							{-m}{\mathrm{e}}%
							{-f}{\mathfrak{e}}%
							{-b}{\bm{e}}%
							{-mb}{\bm{\mathrm{e}}}%
							{+}{E}%
							{+m}{\mathrm{E}}%
							{+b}{\bm{E}}%
							{+mb}{\bm{\mathrm{E}}}%
}[\mbox{\color{red}Use of illegal parameter in command \textbackslash eVar!}]}
\newcommand{\fVar}[1][-]{	\IfStrEqCase{#1}{%
							{-}{f}%
							{-w}{\widetilde{f}}%
							{-m}{\mathrm{f}}%
							{-b}{\bm{f}}%
							{-mb}{\bm{\mathrm{f}}}%
							{+}{F}%
							{+m}{\mathrm{F}}%
							{+b}{\bm{F}}%
							{+mb}{\bm{\mathrm{F}}}%
}[\mbox{\color{red}Use of illegal parameter in command \textbackslash fVar!}]}
\newcommand{\gVar}[1][-]{	\IfStrEqCase{#1}{%
							{-}{g}%
							{-c}{\check{g}}%
							{-m}{\mathrm{g}}%
							{-f}{\mathfrak{g}}%
							{-b}{\bm{g}}%
							{-mb}{\bm{\mathrm{g}}}%
							{-w}{\widetilde{g}}%
							{+}{G}%
							{+m}{\mathrm{G}}%
							{+b}{\bm{G}}%
							{+mb}{\bm{\mathrm{G}}}%
							{+w}{\widetilde{G}}%
}[\mbox{\color{red}Use of illegal parameter in command \textbackslash gVar!}]}
\newcommand{\hVar}[1][-]{	\IfStrEqCase{#1}{%
							{-}{h}%
							{-m}{\mathrm{h}}%
							{-b}{\bm{h}}%
							{-mb}{\bm{\mathrm{h}}}%
							{+}{H}%
							{+m}{\mathrm{H}}%
							{+b}{\bm{H}}%
							{+mb}{\bm{\mathrm{H}}}%
}[\mbox{\color{red}Use of illegal parameter in command \textbackslash hVar!}]}
\newcommand{\qVar}[1][-]{	\IfStrEqCase{#1}{%
							{-}{q}%
							{-c}{\check{q}}
							{-f}{\mathfrak{q}}
							{-w}{\widetilde{q}}%
							{-m}{\mathrm{q}}%
							{-b}{\bm{q}}%
							{-mb}{\bm{\mathrm{q}}}%
							{+}{Q}%
							{+m}{\mathrm{Q}}%
							{+bc}{\bm{\check{Q}}}
							{+b}{\bm{Q}}%
							{+mb}{\bm{\mathrm{Q}}}%
}[\mbox{\color{red}Use of illegal parameter in command \textbackslash qVar!}]}
\newcommand{\uVar}[1][-]{	\IfStrEqCase{#1}{%
							{-}{u}%
							{-m}{\mathrm{u}}%
							{-f}{\mathfrak{u}}%
							{-b}{\bm{u}}%
							{-ub}{\mathrlap{\underline{\bm{u}}}\phantom{\bm{u}}}%
							{-ubw}{\mathrlap{\underline{\bm{\widetilde{u}}}}\phantom{\bm{u}}}%
							{-ubc}{\mathrlap{\underline{\bm{\check{u}}}}\phantom{\bm{u}}}%
							{-mb}{\bm{\mathrm{u}}}%
							{+}{U}%
							{+m}{\mathrm{U}}%
							{+f}{\mathfrak{U}}%
							{+b}{\bm{U}}%
							{+bw}{\bm{\widetilde{U}}}%
							{+bc}{\bm{\check{U}}}%
							{+ub}{\mathrlap{\underline{\bm{U}}}\phantom{\bm{U}}}%
							{+mb}{\bm{\mathrm{U}}}%
}[\mbox{\color{red}Use of illegal parameter in command \textbackslash uVar!}]}
\newcommand{\vVar}[1][-]{	\IfStrEqCase{#1}{%
							{-}{v}%
							{-h}{\hat{v}}%
							{-m}{\mathrm{v}}%
							{-b}{\bm{v}}%
							{-ub}{\mathrlap{\underline{\bm{v}}}\phantom{\bm{v}}}%
							{-ubc}{\mathrlap{\underline{\bm{\check{v}}}}\phantom{\bm{v}}}%
							{-mb}{\bm{\mathrm{v}}}%
							{+}{V}%
							{+m}{\mathrm{V}}%
							{+b}{\bm{V}}%
							{+bc}{\bm{\check{V}}}%
							{+ub}{\mathrlap{\underline{\bm{V}}}\phantom{\bm{V}}}%
							{+mb}{\bm{\mathrm{V}}}%
}[\mbox{\color{red}Use of illegal parameter in command \textbackslash vVar!}]}
\newcommand{\wVar}[1][-]{	\IfStrEqCase{#1}{%
							{-}{w}%
							{-m}{\mathrm{w}}%
							{-b}{\bm{w}}%
							{-ub}{\mathrlap{\underline{\bm{w}}}\phantom{\bm{w}}}%
							{-mb}{\bm{\mathrm{w}}}%
							{+}{W}%
							{+m}{\mathrm{W}}%
							{+b}{\bm{W}}%
							{+ub}{\mathrlap{\underline{\bm{W}}}\phantom{\bm{W}}}%
							{+mb}{\bm{\mathrm{W}}}%
}[\mbox{\color{red}Use of illegal parameter in command \textbackslash wVar!}]}
\newcommand{\xVar}[1][-]{	\IfStrEqCase{#1}{%
							{-}{x}%
							{-w}{\widetilde{x}}%
							{-m}{\mathrm{x}}%
							{-u}{\mathrlap{\underline{x}}\phantom{x}}%
							{-o}{\overline{x}}%
							{-ub}{\mathrlap{\underline{\bm{x}}}\phantom{\bm{x}}}%
							{-b}{\bm{x}}%
							{-mb}{\bm{\mathrm{x}}}%
							{+}{X}%
							{+m}{\mathrm{X}}%
							{+b}{\bm{X}}%
							{+mb}{\bm{\mathrm{X}}}%
}[\mbox{\color{red}Use of illegal parameter in command \textbackslash xVar!}]}
\newcommand{\yVar}[1][-]{	\IfStrEqCase{#1}{%
							{-}{y}%
							{-m}{\mathrm{y}}%
							{-b}{\bm{y}}%
							{-mb}{\bm{\mathrm{y}}}%
							{+}{Y}%
							{-o}{\overline{y}}%
							{+c}{\mathcal{Y}}%
							{+m}{\mathrm{Y}}%
							{+b}{\bm{Y}}%
							{+mb}{\bm{\mathrm{Y}}}%
}[\mbox{\color{red}Use of illegal parameter in command \textbackslash yVar!}]}
\newcommand{\N}[1][-]{		\IfStrEqCase{#1}{%
							{-}{\mathbb{N}}%
							{i}{n}%
							{ii}{m}%
							{iii}{l}%
							{iv}{k}%
							{v}{j}%
							{I}{N}%
							{II}{M}%
							{III}{L}%
							{IV}{K}%
							{V}{J}%
							{sI}{\mathsf{N}}%
							{sII}{\mathsf{M}}%
							{sIII}{\mathsf{L}}%
							{sIV}{\mathsf{K}}%
							{sV}{\mathsf{J}}%
}[\mbox{\color{red}Use of illegal parameter in command \textbackslash \N!}]}
\newcommand{\ket}[1]{|#1\rangle}
\newcommand{\bra}[1]{\langle#1|}
\newcommand{\USet}[1][-]{	\IfStrEqCase{#1}{%
							{-}{\mathrm{U}}%
							{-c}{\mathbb{U}}%
}[\mbox{\color{red}Use of illegal parameter in command \textbackslash USet!}]}
\newcommand{\UpLo}[3]{\mathrlap{#1^{#2}}\phantom{#1}_{#3}}

\def\BibTeX{{\rm B\kern-.05em{\sc i\kern-.025em b}\kern-.08em
    T\kern-.1667em\lower.7ex\hbox{E}\kern-.125emX}}
    
\begin{document}

\title{Feynman Meets Turing:\\ The Curse of Quantum Universality\\
}

\author{%
	Yannik N. B{\"o}ck\hspace{1pt}\(\vphantom{|}^{\ast}\)\qquad%
	Holger Boche\hspace{1pt}\(\vphantom{|}^{\ast}\)\(\vphantom{|}^{\star}\)\qquad%
	Frank H.P. Fitzek\hspace{1pt}\(\vphantom{|}^{\dagger}\)\(\vphantom{|}^{\ddagger}\)
\thanks{\(\vphantom{|}^{\ast}\)\hspace{1pt}Technical University of Munich, School of Computation, Information and Technology,
	Department of Computer Engineering, Chair of Theoretical Information Technology, D-80333 Munich, Germany
	{\tt\small \{yannik.} {\tt\small boeck, boche\}@tum.de}\vspace{0.25\baselineskip}}%
\thanks{\(\vphantom{|}^{\star}\)\hspace{1pt}BMBF Research Hub 6G-life Munich, D-80333; Munich Center for Quantum Science and Technology (MCQST), D-80799; 
	Munich Quantum Valley (MQV) D-80807\vspace{0.25\baselineskip}}%
\thanks{\(\vphantom{|}^{\dagger}\)\hspace{1pt}Telekom Chair of Communication Networks, Technische Universität Dresden (TUD), D-01062 Dresden,
	Germany {\tt\small frank.fitzek} {\tt\small@tu-dresden.de}\vspace{0.25\baselineskip}}%
\thanks{\(\vphantom{|}^{\ddagger}\)\hspace{1pt}BMBF Research Hub 6G-life Dresden, D-01062; TUD Cluster of Excellence “Centre for Tactile Internet with Human-in-the-Loop (CeTI),” D-01062;
	TUD 5G Lab Germany, D-01062\vspace{0.1\baselineskip}}%
}

\maketitle

\begin{abstract}
	We consider a formal model of quantum circuit description languages (QCDLs) in which semantically meaningful programs correspond to computable unitary matrices. 
	We show that any semantically universal QCDL -- that is, any QCDL able to describe all computable unitary matrices, which in turn 
	form the set of matrices we can meaningfully represent on digital hardware -- cannot have a semi-decidable set of semantically meaningful descriptions. 
	In particular, no such language admits a compiler that reliably recognizes all valid program descriptions. This result stands in contrast to classical programming languages. 
	While compilation in languages such as C or C++ may itself involve non-terminating computations, the set of semantically meaningful programs remains recursively enumerable, 
	since successful compilation provides a witness of validity. The essential difference lies in the nature of the semantic domains: classical languages 
	describe partial recursive functions, whereas QCDLs describe total unitary operators. Our analysis establishes a fundamental limitation of quantum circuit description 
	languages and highlights a structural distinction between classical and quantum models of computation at the level of formal language theory.
\end{abstract}

\begin{IEEEkeywords}
	Chomsky hierarchy, quantum circuit description language, quantum circuits, quantum machine language, universal quantum computing.
\end{IEEEkeywords}

\section{Introduction}\label{sec:Introduction}
	In digital computing, formal programming languages provide the interface between the conceptual human understanding of algorithms and
	and executable machine code. A program is typically represented as a finite word over some alphabet, whose syntactic correctness is determined by a formal grammar 
	and whose semantics is given by a mapping into a suitable class of computational objects. This paradigm is well understood in the 
	classical setting, where programming languages describe, for example, Boolean circuits or partial recursive functions.

	A key feature of classical programming languages is that syntactic and semantic validation can be effectively implemented. Compilation 
	itself may involve complex or even non-terminating procedures, as is the case for languages such as C++ with unrestricted template 
	metaprogramming. Nevertheless, successful compilation provides always constitutes a \emph{witness} of the compiler output forming a valid computational object.
	In simple terms, whenever the compiler terminates and signals successful compilation, the compiler's output is definitely a well-formed 
	executable program. Rephrased in the terminology of recursion theory, we say that the set of \emph{semantically meaningful} programs remains \emph{recursively enumerable}:
	albeit not practically feasible or useful, it is theoretically possible to implement a program that sequentially prints all possible successfully compiling source codes.  
	Notably, the recursive enumerability of semantically meaningful programs is closely tied to the fact that classical 
	computing allows for implementing \emph{partial functions}, i.e., functions undefined for some input values. 
	In other words, successful compilation does, in general, not provide any information about the generated program's runtime behavior -- it may or may not terminate
	when executed with a certain input value.

	Quantum circuit description languages (QCDLs), on the other hand, operate over a fundamentally different semantic domain. In the circuit model 
	of quantum computing, programs describe computable unitary matrices, which in turn characterize the time-evolution
	of some quantum system via matrix-vector multiplications. Accordingly, the semantic domain of circuit-model quantum computing 
	does not include a notion of partiality: a matrix is a linear mapping that is well-defined for \emph{all} vectors of a certain dimension; 
	conceptually, it is \emph{impossible} for a matrix-vector multiplication \emph{not} to terminate.
	
	As a heuristic picture, we can think of classical compiling as sidestepping the problem of partiality by moving it to the semantic domain.
	Indeed, the set of \emph{total} programs -- that is, programs that terminate for all possible inputs -- is \emph{not} recursively enumerable in the above sense.
	Since the semantic domain of circuit-model quantum computing does not allow for partiality, one may ask whether a similar phenomenon occurs in this case.  
	In this work, we formalize QCDLs as computable mappings from words over a finite alphabet to computable unitary matrices. Within this framework, 
	we show that semantic universality -- that is, the ability to describe \emph{all} possible computable unitary matrices 
	-- indeed implies recursive \emph{unenumerability}. Given any semantically universal QCDL, there does \emph{not} exist a program
	that sequentially lists all valid descriptions of quantum circuits.	Consequently, given any com-\linebreak piler of any semantically universal 
	QCDL, there must al-\linebreak ways\hfill exist\hfill undetectable\hfill ``garbage\hfill code''\hfill --\hfill source\hfill code\hfill that\hfill suc-\linebreak 
	cessfully\hfill compiles\hfill despite\hfill not\hfill describing\hfill any\hfill quantum\hfill circuit.~~
	
	Our findings reveal a fundamental difference between classical and quantum programming languages at the level of formal language theory. 
	While classical languages admit semi-decidable notions of validity, universal QCDLs necessarily do not. This highlights an intrinsic 
	limitation of quantum circuit description languages that is independent of implementation details or hardware constraints. Notably, 
	our results apply to semantic universality in any \emph{fixed} dimension \(\N[I]\geq 2\).
	Our theory does \emph{not} assume (but is nevertheless valid for) QCDLs that describe quantum circuits of arbitrary dimension.
	
	\subsection{Notation an Terminology}
		For generic sets \(\AA\) and \(\BB\), a \emph{partial function} \(\fVar:\AA\supseteq\rightarrow\BB\) is of the form \(\fVar:\dVar[+](\fVar)\rightarrow\BB\), 
		\(\dVar[+](\fVar)\subseteq\AA\). We call \(\dVar[+](\fVar)\) the \emph{domain} of \(\fVar\). If \(\dVar[+](\fVar) = \AA\), we call \(\fVar\) a \emph{total} 
		function. Since the inclusion \(\dVar[+](\fVar)\subseteq\AA\) includes the \emph{improper} case, every total function is also partial, but not vice versa.
		
		Consider \(\N[I]\in\N\) and generic functions \(\fVar_1:\AA_1\rightarrow\AA_{2},\ldots\) \(\ldots,\fVar_{\N[I]}:\AA_{\N[I]}\rightarrow\AA_{\N[I]+1}\).
		For \(\aVar\in\AA_1\), we define
		\begin{align*}	[\fVar_{\N[I]}\cdots\fVar_{1}](\aVar) := \fVar_{\N[I]}(\cdots\fVar_{1}(\aVar)\cdots),
		\end{align*}
		i.e., \([\fVar_{\N[I]}\cdots\fVar_{1}] :\AA_{1} \rightarrow \AA_{\N[I]+1}\) is the concatenation of the functions \(\fVar_1,\ldots,\fVar_{\N[I]}\).
		
		If an expression \(\mathrm{expr}({\cdot})\) of some form defines the elements of a sequence \((\aVar_{\N[i]})_{\N[i]\in\N}\),
		we employ the notation
		\begin{align*}	(\aVar_{\N[i]})_{\N[i]\in\N} : \aVar_{\N[i]} = \mathrm{expr}(\N[i])
		\end{align*}
		for the sequence's definition. For example, for \(\mathrm{expr}({\cdot}) \equiv \sqrt{{\cdot}}\), we may write 
		\((\aVar_{\N[i]})_{\N[i]\in\N} : \aVar_{\N[i]} = \sqrt{\N[i]}\).

\section{Preliminaries}
	Subsequently, we briefly delineate the relevant background from \emph{recursion theory} and \emph{computable analysis}
	For comprehensive treatments, we refer to~\cite{1987-So,1989-PoRi,2000-Ba,2000-We,2021-Br} and~\cite{2026-BoEA}.
		
	Our analysis builds upon the computing model of Turing machines and \(\mu\)-recursive functions, which we will formally introduce below.		
	\emph{Turing machines} form a mathematical abstraction of digital computing, i.e., they formalize our intuitive 
	understanding of \emph{algorithms}. According to the widely accepted \emph{Church-Turing thesis}~\cite[\emph{Author's Preface}]{1987-So},
	this formalization is definitive: If, in theory, we \emph{cannot} solve a mathematical problem by means of a Turing machine, we definitely 
	\emph{cannot} solve it on a real-world digital computer.

\subsection{Turing Machines and \(\mu\)-recursive Functions} 	
	The set of \emph{\(\mu\)-recursive functions}~\cite{1936-Kl} is the smallest set of partial functions \(\gVar:\N\times\cdots\times\N \supseteq\rightarrow\N\) that
	\begin{itemize}[leftmargin=*]
					\item contains the \emph{successor function}, all \emph{constant functions}, and all \emph{projection functions} and
					\item is closed with respect to \emph{composition}, \emph{primitive recursion}, and \emph{unbounded search}
	\end{itemize}
	(c.f.~\cite[Definition~2.1, p.~8, Definition~2.2, p.~10]{1987-So} for details). 
	We denote the set of \(\mu\)-recursive functions by \(\mathrlap{\fVar[+m]_{\hspace{-0.75pt}\mu}}\phantom{\fVar[+m]}^{\star}\).
	As indicated above, \(\mathrlap{\fVar[+m]_{\hspace{-0.75pt}\mu}}\phantom{\fVar[+m]}^{\star}\) coincides with the set of functions we can compute by means of a 
	\emph{universal Turing machine}~\cite{1936-Tu,1937-Tu-a} (see~\cite{1937-Tu-b} for the proof of the equivalence).
	
	The concept of a \(\mu\)-recursive function's domain has a dedicated interpretation in recursion theory. Computations (in the traditional sense) are 
	ultimately step-wise procedures. Given some input, the relevant procedure may or may not eventually reach a terminal state after a finite number of steps. 
	If it does -- and \emph{only} then --, the computation's result is \emph{well-defined}. Specifically, for \(\AA\subseteq\N\) arbitrary, the 
	informal statement 
	\begin{itemize}[leftmargin=*]
		\item	\emph{there exists an algorithm that, given some input \(\N[i]\in\N\), eventually terminates if and only if we have\(\N[i]\in\AA\)}
	\end{itemize}
	is formally equivalent to the statement
	\begin{itemize}[leftmargin=*]
		\item	\emph{there exists a \(\mu\)-recursive function \(\gVar : \N\supseteq\rightarrow\N\) with domain \(\dVar[+](\gVar) = \AA\).}
	\end{itemize}
	If such a \(\mu\)-recursive function exists, we call the set \(\AA\) \emph{recursively enumerable}. As a consequence of the 
	\emph{halting problem}~\cite[Chapter~I, Section~4, p.~18ff]{1987-So}, the set \(\N\setminus\dVar[+](\gVar)\) is \emph{not} 
	necessarily recursively enumerable. If both \(\AA\) and \(\N\setminus\AA\) are recursively enumerable, we call \(\AA\) a \emph{recursive} set. 
	Note that \(\AA\) is a recursive set if and only if
	\begin{align*} \mathds{1}_{\AA} : 	\N		\rightarrow	\N,~ 
										\N[i]	\mapsto		\begin{cases}	1,	&\text{if}~\N[i]\in\AA,\\
																			0,	&\text{otherwise}
															\end{cases}
	\end{align*}
	-- that is, the \emph{indicator function} of \(\AA\) --, is a \(\mu\)-recursive function. Finally, note that the following statements are \emph{equivalent}:
	\begin{itemize}[leftmargin=*]
		\item	there exists a \(\mu\)-recursive function \(\gVar:\N\supseteq\rightarrow\N\) such that \(\AA = \dVar[+](\gVar)\);
		\item	there exists a (total) \(\mu\)-recursive function \(\hVar:\N\rightarrow\N\) such that we have \(\AA  = \{\hVar(\N[i]) :  \N[i]\in\N\}\).
	\end{itemize}
	For a comprehensive analysis of recursively enumerable sets and their properties (such as the equivalences above), 
	we once more refer to~\cite{1987-So}. 

	In order to perform computations that involve abstract objects such as unitary matrices, we must represent those objects by means
	of our computers native ``information carrier''. In real-world digital computers, any data is ultimately represented by a bit-string in the computers memory. 
	In turn, a bit-string is nothing else than the base-2 expansion of a natural number. The formal study of abstract objects 
	in relation to their natural-number representatives leads to the notion of \emph{modest sets}. A modest set is a pair
	\((\AA_{\mu},\mathsf{N}_{\AA})\) consisting of a generic countable set \(\AA_{\mu}\) and a partial surjective mapping
	\begin{align*}	\mathsf{N}_{\AA} : \N \supseteq \rightarrow \AA_{\mu},
	\end{align*}
	which we refer to as a \emph{numbering} of \(\AA_{\mu}\). Given \(\aVar\in\AA_{\mu}\), a number \(\N[i]\in\N\) that satisfies \(\N[i]\in\dVar[+](\mathsf{N}_{\AA})\) 
	and \(\mathsf{N}_{\AA}(\N[i]) = \aVar\) is called a \emph{realizer} of \(\aVar\). Once suitable modest sets have been established, we can perform computations 
	with their elements through algorithmic manipulation of their realizers. The interested reader may find a comprehensive exposition in~\cite{2000-Ba}.
	
	\pagebreak
	\begin{dfn}[Computable Functions and Sequences]\label{dfn:GenCSq}\mbox{}
		\begin{itemize}[leftmargin=*]
			\item	Let \((\AA_{\mu},\mathsf{N}_{\AA})\) and \((\BB_{\mu},\mathsf{N}_{\BB})\) be modest sets,
				\(\fVar[+] : \AA_{\mu}\supseteq \rightarrow \BB_{\mu}\) any partial mapping, and 
				\(\gVar : \N \supseteq\rightarrow \N\) some \(\mu\)-recursive function that satisfies
				\begin{align*} 	&\N[i]\in\dVar[+](\gVar) ~\text{and}~ [\fVar[+]\mathsf{N}_{\AA}](\N[i]) = [\mathsf{N}_{\BB}\gVar](\N[i])\quad\ldots\\
								&\qquad\ldots\quad \text{for all}~\N[ii] \in \{\N[i]\in\dVar[+](\mathsf{N}_{\AA}):\mathsf{N}_{\AA}(\N[i]) \in \dVar[+](\fVar[+])\}.
				\end{align*}
				We call \(\fVar[+]\) an \emph{\((\mathsf{N}_{\AA},\mathsf{N}_{\BB})\)-computable function}.
			\item	Let \((\AA_{\mu},\mathsf{N}_{\AA})\) be a modest set, \((\aVar_{\N[i]})_{\N[i]\in\N}\subset\AA_{\mu}\) 
				any sequence, and \(\hVar : \N\rightarrow\N\) some total \(\mu\)-recursive function that satisfies
				\begin{align*} 	\aVar_{\N[i]} = [\mathsf{N}_{\AA}\hVar](\N[i]) \qquad \text{for all}~\N[i]\in\N.
				\end{align*}
				We call \((\aVar_{\N[i]})_{\N[i]\in\N}\) an \emph{\(\mathsf{N}_{\AA}\)-computable sequence}.
		\end{itemize}
	\end{dfn}

	Whenever \(\fVar[+]\) is some \((\mathsf{N}_{\AA},\mathsf{N}_{\BB})\)-computable function and \(\mathsf{N}_{\AA}\) and \(\mathsf{N}_{\BB}\)
	are unambiguous from the context, we simply refer to \(\fVar[+]\) as \emph{computable}. The same analogously applies to computable sequences.

	Observe that computability in the sense of Definition~\ref{dfn:GenCSq} is invariant under \emph{Turing equivalence}. Given modest sets \((\AA_{\mu},\mathsf{N}_{\AA,1})\)
	and \((\AA_{\mu},\mathsf{N}_{\AA,2})\), we call \(\mathsf{N}_{\AA,1}\) and \(\mathsf{N}_{\AA,2}\) Turing equivalent if there exist 
	\(\mu\)-recursive functions
	\begin{align*}	\gVar_{1\leftarrow 2},\gVar_{2\leftarrow 1}:\N\supseteq\rightarrow\N
	\end{align*} 
	such that for all \(\N[i]\in\dVar[+](\mathsf{N}_{\AA,1})\) and all \(\N[ii]\in\dVar[+](\mathsf{N}_{\AA,2})\) we have 
	\begin{itemize}[leftmargin=*]
			\item	\(\N[i]\in \dVar[+](\gVar_{2\leftarrow 1})\) and \(\mathsf{N}_{\AA,1}(\N[i]) = {[}\mathsf{N}_{\AA,2}\gVar_{2\leftarrow 1}{]}(\N[i])\);
			\item	\(\N[ii]\in \dVar[+](\gVar_{1\leftarrow 2})\) and \(\mathsf{N}_{\AA,2}(\N[ii]) = {[}\mathsf{N}_{\AA,1}\gVar_{1\leftarrow 2}{]}(\N[ii])\).
	\end{itemize}
	Since \(\mathrlap{\fVar[+m]_{\hspace{-0.75pt}\mu}}\phantom{\fVar[+m]}^{\star}\) is closed under composition, Turing equivalence indeed 
	forms an equivalence relation on the set of numberings of a given set \(\AA_{\mu}\). Given a modest set \((\AA_{\mu},\mathsf{N}_{\AA})\),
	we denote the relevant Turing-equivalence class by \(\UpLo{[\mathsf{N}_{\AA}]}{\mathrm{T}}{\sim}\).
	
	Notably, there exists a unique Turing-equivalence class of \emph{universal \(\mu\)-recursive functions}, which consists of exactly those 
	\(\mu\)-recursive functions that represent \emph{universal Turing machines}. A \(\mu\)-recursive function
	\begin{align*}	\gVar : \N\times\N\rightarrow\N,~(\N[i],\N[ii])\mapsto \gVar(\N[i],\N[ii])
	\end{align*}
	is called universal if, for all \emph{unary} \(\mu\)-recursive functions \(\hVar:\N\supseteq\rightarrow\N\)
	-- that is, \(\mu\)-recursive functions in \emph{one} argument --,
	there exists a \emph{\(\gVar\)-index} \(\N[i]\in\N\) such that for all \(\N[ii]\in\N\), we have 
	\((\N[i],\N[ii])\in\dVar[+](\gVar)\) if and only if we have \(\N[ii]\in\dVar[+](\hVar)\)
	and, for all \(\N[ii]\in\dVar[+](\hVar)\),
	\begin{align*}	\gVar(\N[i],\N[ii]) = \hVar(\N[ii]).
	\end{align*}
	Denoting the set of unary \(\mu\)-recursive functions by \(\UpLo{F}{\star}{\hspace{-0.75pt}\mu,1}\), we may
	then define \(\mathsf{F}_{\smash{\gVar}} : \N \rightarrow \UpLo{F}{\star}{\hspace{-0.75pt}\mu,1}\)
	such that \(\mathsf{F}_{\gVar}(\N[i]) = \hVar\) if and only if \(\N[i]\) is a \(\gVar\)-index of \(\hVar\).
	As follows from the smn\emph{-Theorem}, the equivalence class \(\UpLo{[\mathsf{F}_{\gVar}(\N[i])]}{\mathrm{T}}{\sim}\)
	indeed contains all numberings of \(\UpLo{F}{\star}{\hspace{-0.75pt}\mu,1}\) obtained from universal \(\mu\)-recursive functions
	in this way, and may be extended to \(\UpLo{F}{\star}{\hspace{-0.75pt}\mu}\) via \emph{currying} (see~\cite[Theorem~3.5, p.~16]{1987-So} for details).

	With \(\gVar\) and \(\mathsf{F}_{\gVar}\) as above, we may understand the argument of \(\mathsf{F}_{\gVar}\) -- i.e., the first of the two arguments of \(\gVar\)
	as the theoretical equivalence of executable machine code in a real-world digital general-purpose computer. In coherence
	with our discussion from Section~\ref{sec:Introduction}, \(\mathsf{F}_{\gVar}\) is a \emph{total} numbering: even if for some \(\N[i]\in\N\),
	the function \(\gVar\) is undefined for all \((\N[i],\N[ii])\), \(\N[ii]\in\N\), we find that \(\mathsf{F}_{\gVar}(\N[i])\)
	is still well-defined. Precisely, \(\N[i]\) is a \(\gVar\)-index of the \emph{nowhere-defined function} in this case. 

\subsection{Formal Languages and Chomsky's Hierarchy}
	Throughout this work, we will consider an ordered alphabet \(\mathtt{A} := \{\mathtt{a}_{0},\ldots,\mathtt{a}_{\N[II]}\}\), 
	\(\mathtt{a}_{0} < \ldots < \mathtt{a}_{\N[II]}\), \(\N[II]\in\N\).
	Denoting the \emph{empty word} by \(\varepsilon\), we inductively define the \emph{Kleene closure}
	\begin{align*}	&\W 	:=  \bigcup_{\N[i]\in\N}{\textstyle\mathtt{A}^{\N[i]}},	~\text{where}~
					{\textstyle\mathtt{A}^{0}} := \{\varepsilon\}~\text{and}~\ldots\\
					&\qquad\ldots~{\textstyle\mathtt{A}^{\N[i] +1}} 
						:= 	\begin{cases}	\{(\mathtt{a}) : \mathtt{a}\in\mathtt{A}\},
												&\text{if}~\N[i] = 0,	\\
											\{(\omega,\mathtt{a}) : \mathtt{a}\in\mathtt{A},~\omega \in{\textstyle\mathtt{A}^{\N[i]}}\}, 
												&\text{otherwise}.
							\end{cases}
	\end{align*}
	Consequently, for all \(\omega\in\W\), there exists a unique \(\N[i]\in\N\) -- the \emph{word length} of \(\omega\) -- 
	such that \(\omega \in \mathtt{A}^{\N[i]}\), which we denote by \(|\omega|\).
	The \emph{concatenation} of words is the unique operation \(\circ: \W\times\W \rightarrow \W\), \((\omega_1,\omega_2)\mapsto\omega_1\circ\omega_2\), such that
	\begin{itemize}[leftmargin=*]
			\item for all \(\omega\in\W\), we have \(\varepsilon\circ\omega = \omega\circ\varepsilon = \omega\) 
				(i.e., \(\varepsilon\) is a two-sided \emph{identity element} of \(\circ\));
			\item for all \(\omega_1,\omega_2,\omega_3\in\W\), we have \(\omega_1\circ(\omega_2\circ\omega_3) = (\omega_1\circ\omega_2)\circ\omega_3\)
				(i.e., \(\circ\) is \emph{associative});
			\item for all \(\mathtt{b}_1,\mathtt{b}_2\in\mathtt{A}\), we have \((\mathtt{b}_1)\circ(\mathtt{b}_2) = ((\mathtt{b}_1),\mathtt{b}_2)\).
	\end{itemize}
	Accordingly, \(\W\) is the \emph{free monoid} (see, e.g,~\cite[Section~1.1, p.~1\hspace{2pt}f]{2022-Pe}) on \(\mathtt{A}\). Given \(\N[I]\in\N\), 
	\(\mathtt{b}_1,\ldots,\mathtt{b}_{\N[I]}\in\mathtt{A}\) we will usually write  \(\mathtt{b}_{1}\cdots\mathtt{b}_{\N[I]}\) instead of 
	\(((\ldots(\mathtt{b}_{\smash{1}})\ldots),\mathtt{b}_{\N[I]})\) for ease of notation. Similarly, for \(\omega_1,\omega_2\in\W\), we write 
	\(\omega_1\omega_2\) instead of \(\omega_1\circ\omega_2\). 
	
	Building upon the concatenation of words, we extend the order relation of the alphabet \(\mathtt{A}\) 
	to the monoid \(\W\) as follows.
	\begin{itemize}[leftmargin=*]
			\item For all \(\omega\in\W\), we require \(\omega < \mathtt{a}_0\omega < \ldots < \mathtt{a}_{\N[II]}\omega\).
			\item For all \(\mathtt{b}_1,\mathtt{b}_2\in\mathtt{A}\) and all \(\omega_1,\omega_1\in\W\) that satisfy \(\omega_1 < \omega_2\), 
				we require \(\mathtt{b}_1\omega_1 < \mathtt{b}_2\omega_2\).
	\end{itemize}
	The resulting well-order is unique and induces a total numbering 
	\(\mathsf{N}_{\W}:\N\rightarrow\W\) such that 
	\begin{align*}	\mathsf{N}_{\W}(0) 			= \varepsilon,	\quad
					\mathsf{N}_{\W}(\N[i]+1) 	= \min\big\{\omega\in\W: \mathsf{N}_{\W}(\N[i]) < \omega\big\}
	\end{align*}
	for all \(\N[i]\in\N\). The corresponding equivalence class \(\UpLo{[\mathsf{N}_{\W}]}{\mathrm{T}}{\sim}\) uniquely satisfies the following:
	\begin{itemize}[leftmargin=*]
			\item for all \(\mathsf{N}\in\UpLo{[\mathsf{N}_{\W}]}{\mathrm{T}}{\sim}\), we have \(\mathrm{img}(\mathsf{N}) = \W\);
			\item the functions \((\omega_1,\omega_2)\mapsto\omega_1\omega_2\) and \((\omega_1,\omega_2)\mapsto\mathds{1}_{\smash{\{\omega_1\}}}(\omega_2)\)
				are computable.
	\end{itemize}
	In addition, the word-length function \(|{\cdot}|:\W\rightarrow\N,~\omega\mapsto|\omega|\) is computable likewise. 	

	\begin{definition}	Given the Kleene closure \(\W\), a (\emph{formal}) \emph{language} is a subset \(\LL \subseteq \W\).
	\end{definition}

	We denote the \emph{power set} of \(\W\), by \(\wp(\W)\). Accordingly, \(\LL\) is a formal language if and only if we have \(\LL \in \wp(\W)\). 
	\emph{Chomsky's hierarchy} is a nested family 
	\begin{align*}
		\TT_{3}	\subset \TT_{2} \subset \TT_{1} \subset \TT_{0} \subset \W
	\end{align*}
	of language classes, each of which corresponds to a certain \emph{type} of \emph{formal grammars}. In the literature, the formal grammars that correspond to 
	\(\TT_{\N[i]}\) are usually referred to as \emph{Type-}\(\N[i]\) \emph{grammars}, where \(\N[i]\in\{0,\ldots,3\}\). Specifically,
	\begin{itemize}[leftmargin=*]
		\item a formal grammar is of Type-\(2\) if and only if it is \emph{context-free};
		\item a formal grammar is of Type-\(0\) if and only if it is \emph{unrestricted}.
	\end{itemize}
	Despite being designated as ``unrestricted'', Type-\(0\) grammars do \emph{not} correspond to arbitrary subsets of \(\W\).
	Precisely, unrestricted grammars characterize the set recursively enumerable languages. In other words, we have \(\LL\in\TT_0\) if and only if 
	there exists a total \(\mu\)-recursive function \(\gVar:\N\rightarrow\N\) such that 
	\begin{align*}	\LL = \big\{{[}\mathsf{N}_{\W}\gVar{]}(\N[ii]):\N[ii]\in\N\big\}.
	\end{align*}
	Recursive languages -- that is, languages \(\LL\in\TT_0\) that satisfy \(\W\setminus\LL\in\TT_0\) -- do \emph{not} correspond to any grammar in 
	Chomsky's hierarchy. 
	However, upon defining
	\begin{align*}	\TT_{0.5} := \big\{\LL\in\TT_0 : \W\setminus\LL\in\TT_0\big\},
	\end{align*}
	we obtain \(\TT_{1} \subset \TT_{0.5} \subset \TT_{0}\). 
	Since other details of Chomsky grammars are not of relevance for our analysis, 
	we omit them for brevity. The reader may find a comprehensive introduction to Chomsky's hierarchy in~\cite{2022-Pe}.
	
	\subsection{Computable Analysis}
	For the sake of brevity, we adopt the formal framework of~\cite{2026-BoEA} -- including the relevant notation -- directly.
	\begin{dfn}[Computable Numbers and Computable Unitary Matrices]\mbox{}
		\begin{itemize}[leftmargin=*]
			\item \((\R_{\mu},\mathsf{N}_{\R})\) denotes the modest set of \emph{computable real numbers}~\cite[Definition~3, p.~522]{2026-BoEA};
			\item \((\C_{\mu},\mathsf{N}_{\C})\) denotes the modest set of \emph{computable complex numbers}~\cite[Definition~3, p.~522]{2026-BoEA};
			\item \((\USet_{\mu},\mathsf{N}_{\smash{\USet}})\) denotes the modest set of \emph{computable unitary matrices}~\cite[Definition~4, p.~522]{2026-BoEA}.
		\end{itemize}
		Since the dimension of \(\USet_{\mu}\) is arbitrary throughout all of our analysis, we suppress its explicit mention.
	\end{dfn} 
	
	Informally, a real number is computable if there exists an algorithm that prints the number's decimal expansion up to any desired number of digits.
	Among others, the set of computable real numbers contains all algebraic numbers, and the numbers \(e\) and \(\pi\). Computable complex numbers
	are complex numbers whose real and imaginary parts are computable real numbers. Finally, computable unitary matrices are unitary matrices
	whose entries are computable complex numbers.
	
	The equivalence classes \(\UpLo{[\mathsf{N}_{\R}]}{\mathrm{T}}{\sim}\), \(\UpLo{[\mathsf{N}_{\C}]}{\mathrm{T}}{\sim}\), and
	\(\UpLo{[\mathsf{N}_{\smash{\USet}}]}{\mathrm{T}}{\sim}\) are operationally unique in making the algebraic axioms of
	\(\R\), \(\C\), and \(\USet\) effective. Concerning \(\R\) and \(\C\), we refer to~\cite{AR-1999-He,2021-Br} for details.
	Concerning the computable unitary matrices, we briefly summarize the relevant axioms (details may be found in~\cite{2026-BoEA}):
	\(\UpLo{[\mathsf{N}_{\smash{\USet}}]}{\mathrm{T}}{\sim}\) is the unique equivalence class of numberings such that  
	\begin{enumerate}	[leftmargin=*, itemsep=0.2\baselineskip, topsep=0.2\baselineskip]
						\item\label{item:def::StructEff:GrpOp} 
							the \emph{group operation} \(\mathrm{gop}\vphantom{|}_{\USet} : \USet_{\mu}\times\USet_{\mu} \rightarrow \USet_{\mu}\)
							defined through the matrix multiplication
							\begin{align*}	\uVar[+b]\vVar[+b] =: \mathrm{gop}\vphantom{|}_{\USet}(\uVar[+b],\vVar[+b])
							\end{align*}
							is \(\mathsf{N}_{\smash{\USet}}\)-computable 
							(with \(\dVar[+](\mathrm{gop}\vphantom{|}_{\smash{\USet}}) = \USet_{\mu}\times\USet_{\mu}\));
						\item\label{item:def::StructEff:Norm} 
							the \emph{metric} \(\dVar_{\smash{\USet}} : \USet_{\mu}\times\USet_{\mu} \rightarrow \{\xVar\in\R[-c]:\xVar\geq0\}\)
							defined through the matrix norm
							\begin{align*}	\|\uVar[+b] - \vVar[+b]\| =: \dVar_{\smash{\USet}}(\uVar[+b],\vVar[+b])
							\end{align*}
							is \(\mathsf{N}_{\smash{\USet}}\)-computable 
							(with \(\dVar[+](\dVar_{\smash{\USet}}) = \USet_{\mu}\times\USet_{\mu}\));
						\item\label{item:def::StructEff:Lim} 
							\emph{effective Cauchy limits} 
							\(	((\uVar[+b]_{\N[i]})_{\N[i]\in\N},\nu) 
								\mapsto{\lim}_{\N[i]\to\infty} \uVar[+b]_{\N[i]}
							\)
							are \(\mathsf{N}_{\smash{\USet}}\)-computable,
							where \( (\uVar[+b]_{\N[i]})_{\N[i]\in\N} \subset \USet_{\mu}\) is a computable
							se-\linebreak quence of unitary matrices and \(\nu:\N\rightarrow\N\)
							is a total \(\mu\)-recursive function such that
							\begin{align*}	\Big\|\uVar[+b]_{\N[iii]+\nu(\N[I])} - \uVar[+b]_{\N[iv]+\nu(\N[I])}\Big\|
											<\frac{1}{2^{\N[I]}}
							\end{align*}
							for all \(\N[iii],\N[iv],\N[I]\in\N\);
						\item\label{item:def::StructEff:BsProj} 
							for all \(1\leq\N[iii],\N[iv]\leq\N[I]\), \(\N[I]\in\N\), the \emph{basis-projection func-\linebreak tional}
							\(\beta_{\N[iii],\N[iv]} : \USet_{\mu}(\N[I]) \rightarrow \C[-c]\) defined through
							\begin{align*}	\beta_{\N[iii],\N[iv]}(\uVar[+b]) := \bra{\eVar[-b]_{\N[iii]}}\uVar[+b]\ket{\eVar[-b]_{\N[iv]}}
							\end{align*}
							is \(\mathsf{N}_{\smash{\USet}}\)-computable 
							(with \(\dVar[+](\beta_{\N[iii],\N[iv]}) = \USet_{\mu}(\N[I])\)).			
	\end{enumerate}
	
	Finally, we establish two technical lemmas. These will be required to prove our main contribution:
	Theorem~\ref{thm:noUniversalQCDL} in Section~\ref{sec:Main}. 
	
	\begin{lem}\label{lem:UnitGrpNotEnum}
			Let \(\gamma : [0,1]\cap\R_{\mu} \rightarrow \USet_{\mu}\) be any Banach-Mazur computable
			function that satisfies \(\gamma(0)\neq\gamma(1)\). There exists \(\xVar\in [0,1]\cap\R_{\mu}\) such that
			\(\gamma(\xVar) \in\USet_{\mu}\setminus\{\wVar[+b]_{\N[i]}:\N[i]\in\N\}\).
	\end{lem}
	
	In other words, the computable unitary group is \emph{not} recursively enumerable.
	
	\begin{proof}[Proof of Lemma~\ref{lem:UnitGrpNotEnum}]
		The claim is a modified variant of~\cite[Theorem 1, p.~524]{2026-BoEA}
		obtained through replacing \((\wVar[+b]_{\N[i]})_{\N[i]\in\N}\) in~\cite[Appendix B, Proof of Theorem~1, p.~4\hspace{2pt}ff]{2026-BoEA} 
		by an \emph{arbitrary} computable sequence of unitary matrices.	
	\end{proof}
	
	\begin{lem}\label{lem:Connectedness}
		Let \(\uVar[+b]_0,\uVar[-ub]_1\in\USet_{\mu}\) be arbitrary. There exists a Banach-Mazur computable function
		\(\gamma : [0,1]\cap\R_{\mu} \rightarrow \USet_{\mu}\) that satisfies \(\gamma(0)=\uVar[+b]_0\) and \(\gamma(1) = \uVar[+b]_1\).
	\end{lem}\begin{proof} 
		The claim is an immediate consequence of~\cite[Proposition~1, p.~523\hspace{2pt}f]{2026-BoEA}.
	\end{proof}

\section{The Curse of Quantum Universality}\label{sec:Main}
	Consider a modest set \((\AA_{\mu},\mathsf{N}_{\smash{\AA}})\). A \emph{description language} is a pair \((\LL_{\AA},\fVar[+]_{\AA})\), where 
	\(\LL_{\AA}\subseteq\W\) is a formal language and \(\fVar[+]_{\AA}:\W\supseteq\rightarrow\AA_{\mu}\) is a computable function that satisfies 
	\(\dVar[+](\fVar[+]_{\AA}) \subseteq \LL_{\AA}\). Moreover, we refer to words
	\begin{itemize}[leftmargin=*]
				\item \(\omega\in\LL_{\AA}\) as \emph{syntactically well-formed},
				\item \(\omega\in\W\setminus\LL_{\AA}\) as \emph{syntactically ill-formed},
				\item \(\omega\in\dVar[+](\fVar[+]_{\AA})\) as \emph{semantically meaningful}, and
				\item \(\omega\in\W\setminus\dVar[+](\fVar[+]_{\AA})\) as \emph{semantically meaningless}. 
	\end{itemize}
	The set \(\AA_{\mu}\) is the language's \emph{semantic domain} and, whenever we have \(\mathrm{img}(\fVar[+]_{\AA}) = \AA_{\mu}\), we call \((\LL_{\AA},\fVar[+]_{\AA})\) 
	\emph{semantically universal}. Note that \(\dVar[+](\fVar[+]_{\AA})\) is itself a formal language.\linebreak
	In principle, we could thus restrict our analysis to \(\dVar[+](\fVar[+]_{\AA})\)
	and ignore the larger set \(\LL_{\AA}\) of ``merely'' syntactically well-formed words. The reason for including \(\LL_{\AA}\) in our definition of 
	description languages is a mere practical consideration. As we will see below, source code that is syntactically well-formed but semantically meaningless is
	part of many real-world programming languages. 

	In classical computing, relevant semantic domains include boolean functions and finite state automata -- arguably, VHDL is the most common description language 
	for these domains -- or, in an idealized sense, the \(\mu\)-recursive functions \(\mathrlap{\fVar[+m]_{\hspace{-0.75pt}\mu}}\phantom{\fVar[+m]}^{\star}\).
	For example, assume \(\mathtt{A}\) is the standard ASCII alphabet and consider
	\begin{align*}
		\omega :\equiv \big(~&\texttt{\#include <iostream>}\\
		&\texttt{int main() \{}\\
		&\qquad\texttt{int a, b;}\\
		&\qquad\texttt{std::cin >> a >> b;}\\
		&\qquad\texttt{std::cout << (a + b) << '\textbackslash n';}\\
		&\qquad\texttt{return 0;}\\
		&\texttt{\}}~\big).
	\end{align*}
	According to the C++20-standard~\cite{2020-CPP}, \(\omega\in\W\) is a syntactically well-formed word in the programming language C++. 
	The \(\mu\)-recursive function
	\begin{align*}
	{\textstyle \fVar[+]_{\smash{\text{C++}}}}(\omega) : \mathbb{N}\times\mathbb{N} \to \mathbb{N},~
	(\xVar,\yVar) \mapsto [{\textstyle \fVar[+]_{\smash{\text{C++}}}}(\omega)](\xVar,\yVar) = \xVar + \yVar
	\end{align*}
	is the image of \(\omega\) under the mapping \({\textstyle \fVar[+]_{\smash{\text{C++}}}} : \W \supseteq\rightarrow 
	\mathrlap{\fVar[+m]_{\hspace{-0.75pt}\mu}}\phantom{\fVar[+m]}^{\star}\).
	Accordingly, \(\omega\) is also semantically meaningful. In contrast,
	\begin{align*}
		&\texttt{\#include <iostream>}\\
		&\texttt{int main( \{}\\
		&\qquad\texttt{int a, b;}\\
		&\qquad\texttt{std::cin >> a >> b;}\\
		&\qquad\texttt{std::cout << (a + b) << '\textbackslash n';}\\
		&\qquad\texttt{return 0;}\\
		&\texttt{\}}
	\end{align*}
	is syntactically ill-formed due to mismatched parentheses in the second line, and thus also semantically meaningless.
	Finally, the word
	\begin{align*}
		&\texttt{template<int N>}\\
		&\texttt{struct Loop \{}\\
		&\quad\texttt{static constexpr int value =}~\ldots\\
		&\hspace{2cm}\ldots~\texttt{Loop<N + 1>::value;}\\
		&\texttt{\};}\\
		&\texttt{int x = Loop<0>::value;}
	\end{align*}
	is syntactically well-formed but semantically meaningless. While it conforms to the C++20-standard on a syntactical level,
	it would send an ``unrestricted'' compiler into an infinite loop via unbounded template instantiation. While the C++20-standard requires compiler 
	implementations to limit the maximum depth of template instantiations in some form, the actual limit is compiler specific and thus \emph{not}
	part of the language syntax itself.

	Consider any semantically universal description language \((\LL,\fVar[+])\) whose semantic domain is the \(\mu\)-recursive functions
	\(\mathrlap{\fVar[+m]_{\hspace{-0.75pt}\mu}}\phantom{\fVar[+m]}^{\star}\). In practice, semantically universal programming languages 
	usually satisfy additional syntactical properties. For example, such languages typically provide a constructive method for composing programs.
	Formally we expect that there exists a computable operation
	\begin{align*}	\cVar[-mb]_{\smash{\fVar[+]}} : \W\times\W\rightarrow\W,~(\omega_1,\omega_2)\mapsto \omega_1\hspace{2pt}\cVar[-mb]_{\smash{\fVar[+]}}\hspace{2pt}\omega_2 
	\end{align*}
	such that whenever there exist \(\mu\)-recursive functions \(\fVar_1,\fVar_2:\N\supseteq\rightarrow\N\) that satisfy
	\(\fVar_1 = \fVar[+](\omega_1)\) and \(\fVar_2 = \fVar[+](\omega_2)\), we have 
	\begin{align*}	&\fVar[+](\omega_1\hspace{2pt}\cVar[-mb]_{\smash{\fVar[+]}}\hspace{2pt}\omega_2) : \N\supseteq\rightarrow\N,~\ldots\\
					&\hspace{2cm}\ldots~\N[i]\mapsto {[}\fVar[+](\omega_1\hspace{2pt}\cVar[-mb]_{\smash{\fVar[+]}}\hspace{2pt}\omega_2){]}(\N[i]) = \fVar_1(\fVar_2(\N[i])).
	\end{align*}
	In other words, given the programs \(\omega_1\) and \(\omega_2\) that implement the functions \(\fVar_1\) and \(\fVar_2\), respectively, 
	we do \emph{not} need to ``manually'' find a new program if we want to implement the composition of \(\fVar_1\) and \(\fVar_2\); 
	instead, we can find such a program by means of some canonical method. 

	Syntactical properties of the above kind lead to \emph{effective universality}, which ultimately enable features such as \emph{transpiling}. 
	For our purposes, however, mere semantic universality is of primary relevance. Notably, to facilitate the process of compiling, the formal 
	grammars underlying real-world programming languages is often intentionally kept simple. For example, the ``basic'' structure of the language 
	C belongs to \(\TT_2\). Due to some compiler-related methods relevant for e.g., variable-name resolution, C is ultimately a 
	\(\TT_{0.5}\) language. Given \(\omega\in\W\), there exist three distinct possibilities regarding the compiler behavior:
	\begin{enumerate}[label=\arabic*),leftmargin=*]
		\item the compiler remains in a non-terminating computation, in which case we must have 
			\(\omega\in\LL\setminus\dVar[+](\fVar[+])\);
		\item the compiler exits abnormally and signals that no executable binary was generated, in which case we must have
			\(\omega\in\W\setminus\dVar[+](\fVar[+])\);
		\item the compiler exits normally and returns an executable binary, which is the case if and only if we have
			\(\omega\in\dVar[+](\fVar[+])\).
	\end{enumerate}
	Accordingly, we have \(\dVar[+](\fVar[+])\in\TT_0\): Concerning (compiled) classical programming languages, the set of
	semantically meaningful words is always at least semi-decidable. In many ``simpler'' programming languages, such as, e.g., C,
	we have \(\dVar[+](\fVar[+])\in\TT_{0.5}\) and \(\LL\in\TT_{0.5}\), or \(\LL = \dVar[+](\fVar[+])\).
	For a comprehensive discussion, we refer to~\cite{2006-AhEA}.

	In the circuit model of quantum computing, the semantic domain consists of the computable unitary matrices.
	In the following, let \(\W_{\mathrm{L}}\) be the Kleene closure of some alphabet \(\mathtt{A}_{\mathrm{L}}\)
	and \((\QQ,\qVar[+])\) a QCDL such that
	\begin{align*}	\QQ\subseteq\W_{\mathrm{L}},\quad \qVar[+]:\QQ\supseteq\rightarrow\USet_{\mu}.
	\end{align*}
	The compilation of quantum circuit-description languages is analogous to the classical case, but additionally
	requires an accuracy parameter (c.f. Figure~\ref{fig:QuantumCompiling}). 
	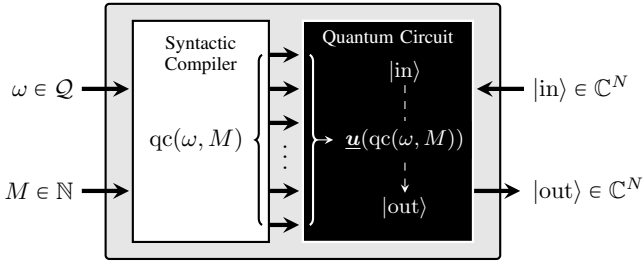
\begin{figure}[!htbp]\centering
		\begin{tikzpicture}[scale =0.9, every node/.style={scale=0.9}]
			\draw (-1.25,-2.5) node [anchor = east, align = right] {\(\omega\in\mathcal{Q}\)};
			\draw (-1.25,-4) node [anchor = east, align = right] {\(\N[II]\in\N\)};
			\draw (5.25,-2.5) node [anchor = west, align = left] {\(|\mathrm{in}\rangle\in\C^{\N[I]}\)};
			\draw (5.25,-4) node [anchor = west, align = left] {\(|\mathrm{out}\rangle\in\C^{\N[I]}\)};
			\draw[thick, fill = black!10!white, rounded corners = 0.5mm] (-0.9,-1.25) rectangle (4.9,-5); 
			\draw[ultra thick, -stealth, shorten >=0.5mm] (-1.25,-2.5) -- (-0.5,-2.5);
			\draw[ultra thick, -stealth, shorten >=0.5mm] (-1.25,-4) -- (-0.5,-4);
			\draw[ultra thick, -stealth, shorten >=0.5mm] (5.25,-2.5) -- (4.5,-2.5);
			\draw[ultra thick, -stealth, shorten >=0.5mm] (4.5,-4) -- (5.25,-4);	
			\draw[thick, fill = white] (-0.5,-1.5) rectangle (1.5,-4.75);
			\draw (0.5,-2) node[anchor = center, align = center,font=\footnotesize] {Syntactic\\ Compiler};
			\draw[ultra thick, -stealth, shorten >=0.5mm] (1.5,-2) -- (2,-2);
			\draw[ultra thick, -stealth, shorten >=0.5mm] (1.5,-2.5) -- (2,-2.5);
			\draw[ultra thick, -stealth, shorten >=0.5mm] (1.5,-3) -- (2,-3);
			\draw[ultra thick, -stealth, shorten >=0.5mm] (1.5,-4) -- (2,-4);
			\draw[ultra thick, -stealth, shorten >=0.5mm] (1.5,-4.5) -- (2,-4.5);
			\draw [thick, decorate, decoration = {brace}] (1.4,-4.5) -- (1.4,-2);
			\draw (1.25,-3.25) node [anchor = east, align = right] {\(\mathrm{qc}(\omega,\N[II])\)};			
			\draw[] (1.75,-3.5) node[rotate = 90, anchor = center, align = center] {\(\cdots\)};
			\draw[thick, draw = white, fill = black] (2,-1.5) rectangle (4.5,-4.75);
			\draw (3.25,-1.75) node[color = white, anchor = center, align = center, font=\footnotesize] {Quantum Circuit};
			\draw [thick, decorate, decoration = {brace}, color = white] (2.1,-2) --  (2.1,-4.5);
			\draw (3.5,-2.25) node [color = white, anchor = center, align = center, font=\small] {\(|\mathrm{in}\rangle\)};
			\draw (3.5,-4.25) node [color = white, anchor = center, align = center, font=\small] {\(|\mathrm{out}\rangle\)};
			\draw[color = white, -stealth, dashed, shorten >= 2mm] (3.5,-3.35) -- (3.5,-4.25);
			\draw[color = white, -stealth, dashed, shorten <= 2mm] (3.5,-2.2225) -- (3.5,-3.15);
			\draw (3.5,-3.25) node [color = white, fill = black, anchor = center, align = center, font=\small] {\(\uVar[-ub](\mathrm{qc}(\omega,\N[II]))\)};
			\draw[color = white, -stealth, shorten >= 5mm] (2.195,-3.25) -- (3,-3.25);	
			\draw[thick, draw = white] (2,-1.5) rectangle (4.5,-4.75);
		\end{tikzpicture}
		\caption{Schematics of Quantum Computing in the Circuit Model.}
		\label{fig:QuantumCompiling}
	\end{figure}
	Formally, let \(\W_{\mathrm{G}}\) be the Kleene closure of some alphabet \(\mathtt{A}_{\mathrm{G}}\) that corresponds to 
	the physical quantum gates available on the hardware under consideration, and denote the unitary transformation that corresponds to some 
	word \(\omega\in\W_{\mathrm{G}}\) by \(\uVar[-ub](\omega)\). That is, if we have \(\mathtt{A}_{\mathrm{G}} = \{\mathtt{a}_0,\ldots
	\mathtt{a}_{\N[II]}\}\) for some \(\N[II]\), the unitary matrices \(\uVar[-ub](\mathtt{a}_0),\ldots,\uVar[-ub](\mathtt{a}_{\N[II]})\in\USet_{\mu}\)
	correspond to the hardware's underlying gate set. Moreover, if \(\omega = \mathtt{b}_1\cdots\mathtt{b}_{\N[I]}\) for \(\mathtt{b}_1,\ldots,
	\mathtt{b}_{\N[I]}\in\mathtt{A}_{\mathrm{G}}\) and some \(\N[I]\in\N\), we have
	\begin{align*}
		\uVar[-ub](\omega) := \uVar[-ub](\mathtt{b}_{\N[I]})\cdots\uVar[-ub](\mathtt{b}_1).
	\end{align*}
	A quantum compiler is then a mapping \(\mathrm{qc} : \W_{\mathrm{L}}\times\N\supseteq\rightarrow\W_{\mathrm{Q}}\)
	that satisfies \(\dVar[+](\mathrm{Q})\times\N\subseteq\dVar[+](\mathrm{qc})\) and
	\begin{align*}	\big\|\mathrm{Q}(\omega) - \uVar[-ub]\big(\mathrm{qc}(\omega,\N[II])\big)\big\| < \frac{1}{2^{\N[II]}}.
	\end{align*}
	for all \(\omega\in\dVar[+](\mathrm{Q}), \N[I]\in\N\). We obtain the following.

	\begin{theorem}\label{thm:noUniversalQCDL}
		If the QCDL \((\QQ,\qVar[+])\) is semantically universal, it satisfies \(\dVar[+](\qVar[+])\in\wp(\W_{\mathrm{L}})\setminus\TT_0\).
	\end{theorem}
	
	\begin{proof}[Proof of Theorem~\ref{thm:noUniversalQCDL}]
		We prove Theorem~\ref{thm:noUniversalQCDL} by contradiction. To this end, subsequently assume there exists a semantically universal QCDL 
		\((\QQ,\qVar[+])\) such that \(\dVar[+](\qVar[+])\in\TT_0\). 
		
		Since \(\TT_0\)-languages are recursively enumerable,
		there exists a total \(\mu\)-recursive function \(\gVar:\N\supseteq\rightarrow\N\) such that we have
		\begin{align*}	\QQ = \big\{{[}\mathsf{N}_{\smash{\W,\mathrm{L}}}\gVar{]}(\N[i]) : \N[i]\in\N\big\}.
		\end{align*}
		Define the sequence \((\omega_{\N[i]})_{\N[i]\in\N} : \omega_{\N[i]} = {[}\mathsf{N}_{\smash{\W,\mathrm{L}}}\gVar{]}(\N[i])\). Consequently, we obtain 
		\begin{align*}	\{\omega_{\N[i]}:\N[i]\in\N\} = \dVar[+](\qVar[+]). 
		\end{align*}
		Moreover, observe that \((\omega_{\N[i]})_{\N[i]\in\N}\) is a computable sequence.
		Per definition of description languages, \(\qVar[+]\) must be a computable function. Thus, 
		\begin{align*}	(\uVar[+b]_{\N[i]})_{\N[i]\in\N} : \uVar[+b]_{\N[i]} = \qVar[+](\omega_{\N[i]})
		\end{align*}
		is a computable sequence of unitary matrices that, due to the semantic universality of \((\QQ,\qVar[+])\), satisfies
		\begin{align*}	\USet_{\mu} = \{\uVar[+b]_{\N[i]}:\N[i]\in\N\}.
		\end{align*}
		Finally, choose \(\gamma : [0,1]\cap\R_{\mu} \rightarrow \USet_{\mu}\) such that it satisfies
		Lemma~\ref{lem:UnitGrpNotEnum} -- by Lemma~\ref{lem:Connectedness}, such a function always exists.
		Then, there exists \(\xVar\in [0,1]\cap\R_{\mu}\) that satisfies
		\begin{align*}	\gamma(\xVar) \in\USet_{\mu}\setminus\{\uVar[+b]_{\N[i]}:\N[i]\in\N\} = \emptyset,
		\end{align*}
		which is absurd.
	\end{proof} 
	
	In other words, there does \emph{not} exists a semantically universal QCDL whose set of semantically meaningful words is semi-decidable.
	Equivalently, there does not exist a program that sequentially prints the semantically meaningful words of any such language. Given any compiler 
	for a semantically universal QCDL, the QCDL's unenumerability also implies the existence of source code that successfully compiles despite not describing 
	any quantum circuit.

	As we have seen above, the compilation process of classical programming languages -- particularly those that are less closely tied to hardware -- may itself be Turing complete. 
	That is, such languages often include compiler directives that do not directly affect the semantics of the generated program, but nevertheless influence 
	the compilation process in a computationally nontrivial way. Prominent examples include macro expansion and template metaprogramming in C++, where the compiler 
	essentially generates additional source code during compilation.

	In such settings, syntactic well-formedness remains a decidable property in the grammar-theoretic sense. However, the correspondence between syntactically 
	well-formed words and semantically meaningful programs breaks down: not every syntactically valid program necessarily gives rise to a semantic object, since 
	compilation itself may fail to terminate. At first glance, this situation appears analogous to the case of quantum circuit description languages (QCDLs).
	There is, however, a fundamental difference. In the classical setting, the discrepancy between syntactic well-formedness and semantic meaningfulness 
	is a consequence of language design, arising from the inclusion of computationally powerful compile-time mechanisms. In principle, one could restrict 
	the language to eliminate such features and recover a setting in which syntactic correctness implies semantic meaningfulness.

	In contrast, for QCDLs this discrepancy is not a matter of design, but an inherent consequence of the structure of the semantic domain.
	This observation necessitates the strict distinction between syntactically well-formed words and semantically meaningful descriptions. 
	In the classical case, the latter correspond precisely to those words for which the compiler terminates successfully and produces an executable output. 
	Consequently, the set of semantically meaningful programs is always recursively enumerable. In the QCDL setting, however, this characterization fails: 
	even if a compiler terminates and signals successful compilation, there is in general no effective procedure to verify that the resulting description 
	indeed corresponds to a valid semantic object.

	In our model of QCDLs, each semantically meaningful word corresponds per definition to some computable unitary matrix.
	Within the QCDL setting (see Figure~\ref{fig:QuantumCompiling}), this definition suggests itself: the very purpose of 
	any QCDL is the description of quantum circuits, which in turn characterize unitary matrices. More complex 
	quantum computing settings will generally involve other 
	components, such as conditional branching upon measurement results, hybrid classical-quantum computations, or partial tracing.
	For quantum programming languages applied in such scenarios, the definition of ``semantically meaningful words''
	may be different. Consider a language \(\QQ\subset\W_{\mathrm{L}}\) and denote the (now arbitrary) set of semantically 
	meaningful words by \(\QQ_{\mathrm{sem}}\). The applicability of our theory then depends on whether we can still effectively assign 
	some computable unitary matrix to each semantically meaningful word -- that is, whether there exists a computable function 
	\(\qVar[+m] :\W_{\mathrm{L}}\supseteq\rightarrow\USet_{\mu}\) such that \(\dVar[+](\qVar[+m]) = \QQ_{\mathrm{sem}}\). In any case, however,
	we can make statements about the intersection \(\QQ_{\mathrm{sem}}\cap\dVar[+](\qVar[+m])\): we either have
	\begin{align*}	\USet_{\mu} \setminus \{\qVar[+m](\omega): \omega \in \QQ_{\mathrm{sem}}\cap\dVar[+](\qVar[+m])\} \neq\emptyset,
	\end{align*}
	i.e., the language is not semantically universal with respect to \(\USet_{\mu}\), or the intersection
	\(\QQ_{\mathrm{sem}}\cap\dVar[+](\qVar[+m])\) is not semi-decidable, i.e., there cannot exist a program 
	that detects whether an arbitrary word \(\omega\in \QQ_{\mathrm{sem}}\) corresponds to some computable unitary matrix or not.

	On the other hand, our analysis is \emph{not} restricted to the circuit model of quantum computing. In fact,
	it applies to any description language whose semantically meaningful words describe computable unitary matrices.
	The same way, it is agnostic to the physical hardware that implements these matrices as physical quantum systems.
	While the circuit model of quantum computing is arguably the most prominent quantum-computing model, it is only
	one of many potential methods to implement arbitrary unitary transformations for computing.
	For example, we may consider hypothetical ``matrix quantum hardware platforms'' that allows for specifying
	the individual entries of the unitary matrix to be implemented directly 
	(see Figure~\ref{fig:QuantumComputerFull} for an exemplary setup). As long as this hardware follows the traditional
	``quantum data, digital control'' paradigm -- that is, as long as the hardware is embedded in a digital environment
	environment that handles the quantum system's initialization, process control, and readout -- we must 
	specify the quantum algorithm to be executed digitally. Whenever this specification allows us to 
	digitally compute the entries of the resulting unitary matrix, our theory applies.

		\begin{figure}
		\begin{tikzpicture}
			\draw[thick, fill = black!10!white, rounded corners = 0.5mm] (-1.55,0.5) rectangle (4.35,-6.5); 
			\draw[thick, fill = white] (-1.2,0.25) rectangle (4,-1);
			\draw (1.4,0.25) node[anchor = north, font=\footnotesize] {Matrix Synthesis};
			\draw[ultra thick, -stealth, shorten >=0.5mm] (1.4,-1) -- (1.4,-1.15) -- (0.875,-1.65) -- (0.875,-2);
			\draw[ultra thick, -stealth, shorten >=0.5mm] (3,-1) -- (3,-1.5);
			\draw[ultra thick, stealth-stealth, shorten >=0.5mm, shorten <=0.5mm] (-0.2,-1) -- (-0.2,-4.65) -- (0.45,-5.15) -- (0.45,-5.5);
			\draw[ultra thick] (-0.2,-4.65) -- (-0.85,-5.15) -- (-0.85,-5.5);
			\draw(-0.2,-1) node[anchor=south, align=center,font=\scriptsize] {Post-Proc.\\BUS};
			\draw(1.4,-1) node[anchor=south, align=center,font=\scriptsize] {ME Sequence};
			\draw(3,-1) node[anchor=south, align=center,font=\scriptsize] {Initial-State\\Selection};
			\draw[thick, fill = white] (0.25,-2) rectangle (1.5,-4.5);
			\draw (0.5625,-3.25) node[rotate = 270, anchor = center, align = center,font=\footnotesize] {ME Register};
			\draw[ultra thick, -stealth, shorten >=0.5mm] (1.5,-2.5) -- (2,-2.5);
			\draw[ultra thick, -stealth, shorten >=0.5mm] (1.5,-3) -- (2,-3);
			\draw[ultra thick, -stealth, shorten >=0.5mm] (1.5,-4) -- (2,-4);
			\draw [thick, decorate, decoration = {brace}] (1.4,-4) -- (1.4,-2.5);
			\draw (1.3,-3.25) node [anchor = east, align = right, scale = 0.75] {\(\begin{matrix}\theta_1\\\raisebox{2.5pt}{\(\vdots\)}\\\theta_{\N[I]}\end{matrix}\)};			
			\draw[] (1.75,-3.5) node[rotate = 90, anchor = center, align = center] {\(\cdots\)};
			\draw[ultra thick, -stealth, shorten >=0.5mm] (3,-5) -- (3,-5.5);
			\draw[thick, draw = white, fill = black] (2,-1.5) rectangle (4,-5);
			\draw[thick, draw = white] (2,-2) -- (4,-2);
			\draw (3,-1.75) node[color = white, anchor = center, font=\scriptsize] {Preparation};
			\draw[thick, draw = white] (2,-4.5) -- (4,-4.5);
			\draw (3,-4.75) node[color = white, anchor = center, font=\scriptsize] {Measurement};
			\draw (3.65,-3.25) node[rotate = 270, color = white, anchor = center, align = center, font=\footnotesize] {Quantum Processor};
			\draw [thick, decorate, decoration = {brace}, color = white] (2.1,-2.5) --  (2.1,-4);
			\draw (3,-2.25) node [color = white, anchor = center, align = center, font=\small] {\(|\mathrm{init}\rangle\)};
			\draw (3,-4.25) node [color = white, anchor = center, align = center, font=\small] {\(|\mathrm{out}\rangle\)};
			\draw[color = white, -stealth, dashed, shorten >= 2mm] (3,-3.35) -- (3,-4.25);
			\draw[color = white, -stealth, dashed, shorten <= 2mm] (3,-2.2225) -- (3,-3.15);
			\draw (3,-3.25) node [color = white, fill = black, anchor = center, align = center, scale = 0.75] {\(\uVar[-ub](\theta_{1:\N[I]})\)};
			\draw[color = white, -stealth, shorten >= 5mm] (2.195,-3.25) -- (3,-3.25);	
			\draw[thick, fill = white] (-0.55,-5.5) rectangle (4,-6.25);
			\draw (1.725,-5.875) node[anchor = center, font=\footnotesize] {Classical Post-Processing};
			\draw[ultra thick] (-0.55,-5.875) -- (-0.85,-5.875) -- (-1.125,-5.6); 
			\draw[ultra thick] (-0.85,-5.885) node [anchor = center, align = center] {\(\bullet\)}; 
			\draw[ultra thick, -stealth, shorten >=1.2cm] (-1.235,-5.875) -- (-1.65,-5.875) -- (-3.075,-5.25);
			\draw(-0.2,-2.825) node[anchor=north, align=center,font=\scriptsize, rotate = 270] 	
				{	Post-Proc. Instrunctions \(\rightarrow\)\\
					\(\leftarrow\) Repeat Computation?
				};
			\draw[ultra thick, -stealth, shorten >=0.5mm] (-3.075,-0.75) -- (-1.65,-0.375) -- (-1.2,-0.375);
			\draw[ultra thick, -stealth, shorten >=0.785cm, densely dotted] (-3.075,-5.25) -- (-3.075,-3);
			\draw[ultra thick, -stealth, shorten >=0.785cm, densely dotted] (-3.075,-3) -- (-3.075,-0.75);
			\draw[draw = black, fill = white!90!black] (-3.075,-0.75) ellipse (1.35cm and 0.75cm);
			\draw[draw = black, fill = white!90!black] (-3.075,-3) ellipse (1.35cm and 0.75cm);
			\draw[draw = black, fill = white!90!black] (-3.075,-5.25) ellipse (1.35cm and 0.75cm);
			\draw[] (-3.075,-0.75) node [anchor = center, align = center, font=\footnotesize] {Mathematical\\ Description\\ of Objects};
			\draw[] (-3.075,-3) node [anchor = center, align = center, font=\footnotesize] {Real-World\\ Physical Problem};
			\draw[] (-3.075,-5.25) node [anchor = center, align = center, font=\footnotesize] {Computed Solution\\ to Physical Problem};
		\end{tikzpicture}
		\caption{\emph{Schematics of a ``Matrix Quantum Computer''}. The matrix-synthesis handler receives the description of a quantum algorithm intended 
			to solve some real-world physical problem. The description employs a high-level quantum programming language
			that hides the underlying computations from the user. Based on the given description, the handler 
			determines suitable post-processing steps, selects an initial state for the available qubit memory, and extracts
			a vector of \emph{matrix elements} (ME) \(\theta_1,\ldots,\theta_{\N[I]} \equiv \theta_{1:\N[I]}\) that correspond to the
			\(\N[I] = \N[II]\cdot\N[II]\) entries of a unitary \(\N[II]\)-by-\(\N[II]\) matrix \(\uVar[-ub](\theta_{1:\N[I]})\). The quantum matrix processor
			reads the entries from the ME register and implements the corresponding unitary matrix as the evolution of a physical quantum system.}
			\label{fig:QuantumComputerFull}
		\end{figure}
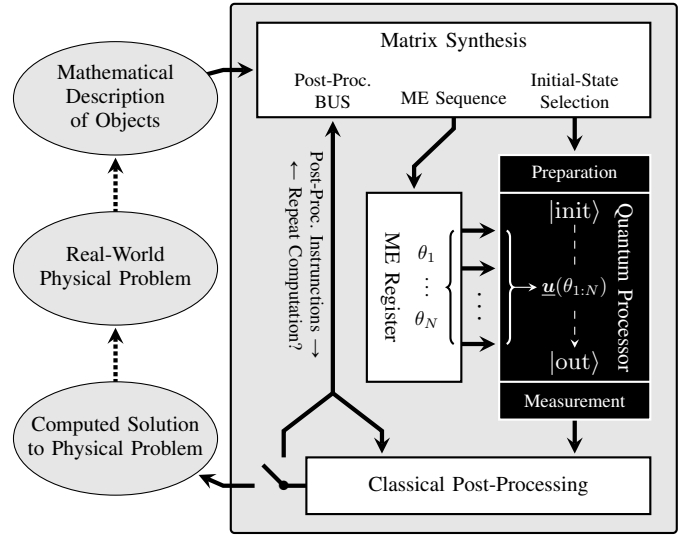
	
\section{Conclusion}
	We have shown that semantically universal quantum circuit description languages exhibit an inherent limitation: 
	the set of semantically meaningful program descriptions is not recursively enumerable. Consequently, no compiler can semi-decide the validity of 
	quantum circuit descriptions in such languages.

	This stands in sharp contrast to classical programming languages. There, even in the presence of non-terminating compilation procedures, 
	the set of valid programs remains recursively enumerable. The underlying reason is that classical semantic domains admit partiality, 
	allowing programs to denote undefined computations. In the quantum setting, however, semantic objects are total unitary operators, and 
	validity requires that a description corresponds to such an object.

	From a computability-theoretic perspective, this places universal QCDLs in a similar category as description systems for total recursive functions, 
	whose index sets are known to be non-recursively enumerable. The result thus identifies a fundamental structural distinction between classical and 
	quantum programming paradigms.

\section*{Acknowledgment}
The authors used OpenAI's \emph{ChatGPT} for language and grammar editing. All content was reviewed and edited by the authors, who take full 
responsibility for the final work.

\balance
\bibliographystyle{IEEEtran}
\bibliography{IEEE-QuWk-2026.Bibliography.bib}

\vspace{12pt}
\end{document}